\RequirePackage{fix-cm}
\documentclass[smallextended,envcountrest]{svjour3}       
\smartqed  
\usepackage{graphicx}
\usepackage{latexsym}
\usepackage{amssymb}
\usepackage{amsmath} %
\usepackage{amscd} %
\usepackage{mathrsfs} 
\numberwithin{equation}{section}
\newcommand{\cstar}{{{C}}^{\ast}}%
\newcommand{\R}{{\mathbb{R}}}%
\newcommand{\Z}{{\mathbb{Z}}}%
\newcommand{\NN}{{\mathbb{N}}}%
\newcommand{\id}{{\mathbf{1}}}
\newcommand{\unit}{1}
\newcommand{\zeromap}{{\mathbf{0}}}
\newcommand{\Bl}{{\mathfrak B}}
\newcommand{\Znu}{\Z^{\nu}}%
\newcommand{\vp}{\varphi}
\newcommand{\ome}{\omega}
\newcommand{\I}{{\mathrm{I}}}%
\newcommand{\J}{{\mathrm{J}}}%
\newcommand{\Al}{\mathcal{A}}%
\newcommand{\Ale}{\Al_{+} }%
\newcommand{\Alo}{\Al_{-}}%
\newcommand{\Ali}{{\Al}(\{i\})}%
\newcommand{\AlI}{{\Al}({\I})}%
\newcommand{\AlJ}{{\Al}({\J})}%
\newcommand{\AlIe}{\AlI_{+}}%
\newcommand{\AlIo}{\AlI_{-}}%
\newcommand{\Fp}{F_{+}}%
\newcommand{\Fm}{F_{-}}%
\newcommand{\Gp}{G_{+}}%
\newcommand{\Gm}{G_{-}}%
\newcommand{\al}{\alpha}
\newcommand{\alt}{\alpha_{t}}%
\newcommand{\Hil}{\mathscr{H}}%
\newcommand{\Qseq}{\mathscr{Q}}%
\newcommand{\OPERA}{\hat{\mathcal{O}}}%
\newcommand{\Ome}{\Omega}%
\newcommand{\Hilome}{{\Hil}_{\ome}}%
\newcommand{\piome}{\pi_{\ome}}%
\newcommand{\Omeome}{\Ome_{\ome}}%
\newcommand{\aicr}{a_i^{\ast}}%
\newcommand{\ai}{a_i}%
\newcommand{\ajcr}{a_j^{\ast}}%
\newcommand{\aj}{a_j}%
\newcommand{\core}{\Al_{\circ}}%
\newcommand{\coree}{{\core}_+}%
\newcommand{\coreo}{{\core}_-}%
\newcommand{\Qast}{Q^{\ast}}%
\newcommand{\cha}{{\Psi}}%
\newcommand{\del}{\delta}
\newcommand{\chanic}{\cha}%
\newcommand{\chanicast}{\chanic^{\;\ast}}
\newcommand{\delchanic}{\del_{\chanic}}%
\newcommand{\delchanicast}{\del_{\chanic}^{\;\;\;\ast}}%
\newcommand{\alchanict}{\al^{\!\chanic}_{\,t}}%
\newcommand{\CA}{{\cal{C}}}%
\newcommand{\Derchanic}{d_{\chanic}}
\newcommand{\twoisan}{\bigl\{ 2i-1,\, 2i,\, 2i+1 \bigr\}}
\newcommand{\twoitwoj}{\bigl\{(2i-1,2j),\, (2i,2j-1),\, (2i,2j),\, (2i+1,2j),\,
(2i,2j+1)\bigr\}}
\newcommand{\twoij}{\J_{(2i,2j)}}
\newcommand{\Altwoijo}{\Al\bigl(\twoij\bigr)_{-}}
\newcommand{\chanictwoisan}{\chanic\bigl(\twoisan\bigr)}
\newcommand{\chanictwoisanast}{\chanic\bigl(\twoisan\bigr)^{\ast}}
\newcommand{\chanictwoij}{\chanic\bigl(\twoij\bigr)}
\newcommand{\atwokcr}{a_{2k}^{\, \ast}}%
\newcommand{\atwok}{a_{2k}}%
\newcommand{\atwokpcr}{a_{2k+1}^{\, \ast}}%
\newcommand{\atwokp}{a_{2k+1}}%
\newcommand{\atwolcr}{a_{2l}^{\, \ast}}%
\newcommand{\atwol}{a_{2l}}%
\newcommand{\atwolmcr}{a_{2l-1}^{\, \ast}}%
\newcommand{\atwolm}{a_{2l-1}}%
\newcommand{\atwoitwojcr}{a_{(2i, 2j)}^{\, \ast}}%
\newcommand{\atwoimtwoj}{a_{(2i-1, 2j)}}
\newcommand{\atwoitwojm}{a_{(2i, 2j-1)}}
\newcommand{\atwoiptwoj}{a_{(2i+1, 2j)}}
\newcommand{\atwoitwojp}{a_{(2i, 2j+1)}}
\newcommand{\Qnic}{Q_{\rm{Nic}}}
\newcommand{\Qnicast}{{\Qnic}^{\!\!\!\ast}}
\newcommand{\Qnicastsq}{{\Qnic}^{\!\!\!\ast 2}}
\newcommand{\Hnic}{H_{\rm{Nic}}}
\newcommand{\bra}[1]{\langle #1|}
\newcommand{\ket}[1]{|#1 \rangle}
\newcommand{\Utome}{U_{\ome}(t)}%
\newcommand{\Home}{H_{\ome}}%
\newcommand{\Ikl}{{\rm{I}}_{[2k,2l]}}
\newcommand{\Izeroone}{{\rm{I}}_{[0,2]}}
\newcommand{\Izerotwo}{{\rm{I}}_{[0,4]}}
\newcommand{\Izerothree}{{\rm{I}}_{[0,6]}}
\newcommand{\Izerolm}{{\rm{I}}_{[0,2l]\times [0,2m]}^{\,2}}
\newcommand{\zetai}{\zeta_{i}}
\newcommand{\opkappai}{\hat{\kappa}_{i}}
\newcommand{\opkappaiminus}{\hat{\kappa}_{i-1}}
\newcommand{\opkappaiplus}{\hat{\kappa}_{i+1}}
\newcommand{\IDplkl}{r^{+}_{[2k,2l]}}
\newcommand{\IDmikl}{r^{-}_{[2k,2l]}}
\newcommand{\IDplzeroone}{r^{+}_{[0,2]}}
\newcommand{\IDmizeroone}{r^{-}_{[0,2]}}
\newcommand{\IDplzerotwo}{r^{+}_{[0,4]}}
\newcommand{\IDmizerotwo}{r^{-}_{[0,4]}}
\newcommand{\IDplzerothree}{r^{+}_{[0,6]}}
\newcommand{\IDmizerothree}{r^{-}_{[0,6]}}
\newcommand{\Imtil}{\widetilde{\rm{I}}_{[-m-1,m+1]}}
\journalname{J. Stat. Phys.}
\begin{document}

\title{Ergodicity breaking and Localization
 of  the Nicolai supersymmetric  fermion lattice  model
}

\titlerunning{Ergodicity breaking and Localization
 of  the Nicolai SUSY fermion lattice model}        

\author{Hajime Moriya}


\institute{H. Moriya \at
              Kanazawa University
College of Science and Engineering. Japan.
}

\date{July 4 2018}

\maketitle

\begin{abstract}
We investigate   dynamics  of a supersymmetric  fermion  lattice model  
   introduced by Nicolai [J. Phys. A. Math. Gen. {{9}} (1976)].
 We show that the Nicolai model has   infinitely many   local constants of motion 
for its  Heisenberg time evolution,  and therefore  ergodicity   
 (with respect to thermal equilibrium  states) breaks.
  It has   infinitely many degenerated classical 
 ground states. This phenomena is  considered  as localization  at zero temperature.
 From  a  viewpoint of  perturbation theory,  we explain
  why delocalization  is suppressed 
at zero temperature despite  its  disorder-free  translation-invariant quantum interaction.
\end{abstract}

\section{Introduction}
\label{sec:INTRO}
 We consider a spinless fermion lattice model proposed  by Nicolai  \cite{NIC}.
 The Nicolai model    satisfies the  same algebraic relation 
 as  ${{\cal{N}}}=2$  supersymmetry \cite{WEIN},  
although  it consists of only fermions without   bosons.
 Using  a  general  formulation   given in \cite{AHP16}  
  we formulate  the Nicolai model as a supersymmetric $\cstar$-dynamical system. 
We investigate its dynamical properties from a rigorous 
   $\cstar$-algebraic approach, see e.g.  \cite{BRII} \cite{THIR4}.

We show that the  Nicolai model exhibits several  non-ergodic  properties.
 It has  infinitely many local constants of motion that are frozen  under  the Heisenberg 
 time evolution  defined on the  infinitely extended  system. 
These local constants of motion 
 are all fermionic, and  the  number of them increases exponentially
 with respect to the volume of  subsystems.  
We can  readily  show   breaking ergodicity  for the Nicolai model
  from  such   local constants of motion.
Precisely, the ergodicity  as defined by  Mazur  \cite{MAZUR}  
 is broken for all  KMS (thermal equilibrium) states \cite{HHW}.
 Furthermore  we show that there exist   infinitely many classical  ground states  on the  Fock space.
 Those highly degenerated ground states   can be  considered  as  localization phenomena 
 at zero temperature.

It is widely believed that  ergodicity  breakdown is caused by integrability. 
We shall investigate   ``quantum integrability'' that   the  Nicolai model possesses.
In  \cite{CAUX}  Caux-Mossel proposed  a new characterization  of 
 quantum integrability intended for  many-body quantum dynamics.
In their language, the Nicolai model  is categorized to  the class called  ``constant quantum   integrable''. 
 This  class includes 
 non-interacting (free) fermion models as its  typical example, and it  
  is more integrable than  the  
 ``linear-quantum-integrable class''  to which the Heisenberg spin chain   belongs.

 As  mentioned above,  the Nicolai model exhibits   
certain  many-body localization  at zero temperature,  although there is  no disorder. 
We  shall recall a  general   mechanism of delocalization  
 for  quantum many-body dynamics
 proposed  by De Roeck-Huveneers \cite{Deloc}.  It  
 is  based on  perturbation argument under   certain 
  assumptions upon   translation-invariant Hamiltonians.
   The Nicolai model gives  an exceptional disorder-free quantum many-body Hamiltonian  
to which the  scenario  of  \cite{Deloc}  can not  apply.
 The precise statement will be given in   the main text.

For concreteness, we   deal with  the  Nicolai model 
  on one-dimensional integer lattice.
 However, all   statements shown in this paper 
 are valid  for  the  Nicolai model defined on any dimensional integer lattice 
with some obvious  modification.

\section{Nicolai supersymmetric fermion lattice model}
\label{sec:SUSY-FERMION}

\subsection{Supersymmetry}
\label{subsec:SUSY}
This subsection   provides  a brief summary of supersymmetry (SUSY).
  For  general references of supersymmetry see e.g. \cite{WEIN}.

 Let  $F$ be   a  positive  operator on a Hilbert space $\Hil$
 whose eigenvalues are non-negative integers.
Then the  Hilbert space has a graded structure $\Hil=\Hil_{+}\oplus  \Hil_{-}$,
where   $(-1)^{F}$ 
has  eigenvalue $+1(-1)$ for any vector in  
$\Hil_{+}(\Hil_{-})$, respectively.
 Consider  a conjugate pair of   linear operators $Q$ and $\Qast$ on $\Hil$,
 where $\ast$ denotes the  adjoint of linear operators.
Assume that they  are  fermionic,   
\begin{equation}
\label{eq:Qodd-intro}
\{(-1)^{F},\ Q\}= \{(-1)^{F},\ Q^{\ast}\}=0. 
\end{equation}
Assume further  that they  are  nilpotent,     
\begin{align}
\label{eq:Qnil-intro}
Q^{2}=0={Q}^{\ast\,2}.
\end{align}
We define  the  Hamiltonian as   
\begin{align}
\label{eq:Hsusy}
H:=\{ Q,\; {Q}^{\ast} \}\equiv Q {Q}^{\ast}+{Q}^{\ast}Q.
\end{align}
From  \eqref{eq:Hsusy} and  \eqref{eq:Qnil-intro}
 we see that  
\begin{equation}
\label{eq:HQkakan}
[H,\ Q]=[H,\ Q^{\ast}]=0.
\end{equation}  
The   algebraic structure  satisfied by   
 $\{Q,\, {Q}^{\ast},\, H,\, (-1)^N,\,\Hil\}$ 
 is called   ${\cal{N}}=2$ supersymmetry.
  The Nicolai model which will be introduced  in this section  is a
  fermion lattice model that has  ${\cal{N}}=2$  supersymmetry.

\subsection{Fermion lattice system}
\label{subsec:FERSYSTEM}
We  recall  a general $\cstar$-algebraic  formulation of  fermion lattice systems
 by which we will  provide   precise formulation of supersymmetric dynamics 
 of  the  Nicolai model.
We consider   integer lattice  $\Znu$ of  any  $\nu\in\NN$.
 For any subset $\I$ of $\Znu$ we denote the number of sites in $\I$ by  $|\I|$. 
 The notation `$\I \Subset \Znu$' means that a subregion  $\I\subset\Znu$ 
 contains  finite number of sites  in it.

We consider interacting  spinless fermions  over  $\Znu$.
Let $\ai$ and $\aicr$ denote the annihilation  
operator and the creation operator of a spinless fermion at $i\in \Znu$, respectively. 
Those obey  the canonical anticommutation relations (CARs):   
\begin{align}
\label{eq:CAR}
\{ \aicr, \aj \}&=\delta_{i,j}\, \unit, \nonumber \\
\{ \aicr, \ajcr \}&=\{ \ai, \aj \}=0.
\end{align}
For each site  $i\in \Znu$  the fermion number operator is defined by
\begin{align}
\label{eq:ni}
n_i:=\aicr \ai.
\end{align}
A formal  infinite sum  $F:=\sum_{i\in\Znu}n_i$ will denote 
 the total fermion number operator.

For each $\I\Subset\Znu$, 
 $\AlI$ denotes  the  finite-dimensional algebra
generated by $\{\ai, \, \aicr\, ;\;i\in \I\}$.
 For   $\I \subset \J\Subset \Znu$, $\AlI$ is naturally 
 imbedded into $\AlJ$ as a subalgebra. We define the local algebra as
\begin{equation} 
\label{eq:CARloc}
\core:=\bigcup_{\I \Subset \Znu }\AlI.
\end{equation}
Taking the  norm completion of the normed $\ast$-algebra $\core$
 we obtain a   $\cstar$-algebra $\Al$ that is  called  the CAR algebra.

Let $\gamma$ denote the automorphism on the $\cstar$-algebra 
 $\Al$ determined  by    
 \begin{equation}
\label{eq:CARgamma}
\gamma(\ai)=-\ai, \quad \gamma(\aicr)=-\aicr,\quad \forall i\in \Znu.
\end{equation}
The grading automorphism $\gamma$ is heuristically given by $\rm{Ad}(-1)^F$.
Obviously,   $\gamma\circ \gamma={\text{id}}$. 
The  system $\Al$ is decomposed into the even part and the odd part:  
\begin{align}
\label{eq:grad}
\Al&=\Ale\oplus \Alo,\quad 
\Ale= \{A\in \Al| \; \gamma(A)=A\},\quad
\Alo= \{A\in \Al|\;  \gamma(A)=-A\}.
\end{align}
Similarly,  for each  $\I\Subset\Znu$ we  consider a natural  graded  structure,  
\begin{equation}
\label{eq:CARIeo}
\AlI=\AlIe\oplus\AlIo,\ \ \text{where}\ \ 
 \AlIe := \AlI\cap \Ale,\quad  \AlIo := \AlI\cap \Alo,
 \end{equation}
and for the local algebra 
\begin{equation}
\label{eq:core-grad}
\core=\coree\oplus\coreo,\ \ \text{where}\ \ 
 \coree := \core \cap \Ale,\quad  \coreo := \core\cap \Alo.
 \end{equation}
The  graded commutator  on the graded algebra $\Al$ is  defined as  
\begin{align}
\label{eq:gcom}
[\Fp, \;  G]_{\gamma} &= [\Fp, \;  G]
  {\mbox {\ \ for \ }}
\Fp \in \Ale, \ G \in \Al, \nonumber\\ 
[\Fm, \; \Gp]_{\gamma} &= [\Fm,  \; \Gp] {\mbox {\ \ for \ }}
\Fm \in \Alo, \ \Gp \in \Ale, \nonumber\\
[\Fm, \; \Gm]_{\gamma} &= \{\Fm,  \; \Gm\} {\mbox {\ \ for \ }}
\Fm \in \Alo, \ \Gm \in \Alo. 
\end{align}
By  the  CARs 
 \eqref{eq:CAR}   the  $\gamma$-locality holds:  
\begin{equation}
\label{eq:glocality}
[A,\; B]_{\gamma}=0{\text {\ \ for every  \ }} A \in \AlI 
{\text {\ and \ }}B\in \AlJ {\text {\ \ if   \ }} \I\cap\J=\emptyset,\ \I,\J\Subset\Znu.
\end{equation}

We  introduce  some basic  transformations on  the fermion lattice system. 
 Let $\sigma$ denote the shift-translation  automorphism group on $\Al$.
 Namely for   $k\in\Znu$
\begin{align}
\label{eq:sigk}
\sigma_{k} (a_{i})=a_{i+k},\quad 
\sigma_{k} (a_{i}^{\ast})=a_{i+k}^{\ast},
\quad \forall i  \in \Znu.
\end{align}
By $\gamma_{\theta}$ 
$(\theta \in [0,\; 2\pi])$
we denote  the global $U(1)$-symmetry  defined as   
\begin{align}
\label{eq:U1}
\gamma_{\theta} (a_{i})=e^{-i\theta}a_{i},\quad 
\gamma_{\theta} (a_{i}^{\ast})=e^{i\theta}a_{i}^{\ast},
\quad \forall i  \in \Znu.
\end{align}
By definition  $\gamma_{\pi}$ is equal to  the grading $\gamma$
  of  \eqref{eq:CARgamma}.
We may consider 
the  particle-hole transformation $\rho$: 
\begin{align}
\label{eq:particle-hole}
\rho (\ai)=\aicr,\quad 
\rho (\aicr)=\ai,
\quad \forall i  \in \Z.
\end{align}

\subsection{The Nicolai  model}
\label{subsec:NIC}
We  investigate dynamical properties of 
 a supersymmetric fermion lattice model 
  given by Nicolai in  \cite{NIC}.
 We shall introduce the Nicolai model on  integer lattice $\Z$. 
However,  one  can easily extend the Nicolai  model 
 to  any   dimensional integer lattice  $\Znu$.
 All  statements which we will  show are   valid 
   for the multi-dimensional case,    
 see $\S$\ref{subsec:Multi} for the detail.

Let us consider 
the following  formal infinite sum of fermion operators  
\begin{align}
\label{eq:Qnic}
\Qnic:=
\sum_{i \in \Z}  
\chanic(\{2i-1, 2i, 2i+1\}),\quad 
\chanic(\{2i-1, 2i, 2i+1\}):= a_{2i+1} a^{\ast}_{2i} a_{2i-1}.
\end{align}
Then 
\begin{align}
\label{eq:Qnicast}
\Qnicast=\sum_{i \in \Z}  
\chanic(\{2i-1, 2i, 2i+1\})^{\ast},\quad 
\chanic(\{2i-1, 2i, 2i+1\})^{\ast}
= a_{2i-1}^{\ast} a_{2i} a_{2i+1}^{\ast}.
\end{align}
Those are fermionic as in  \eqref{eq:Qodd-intro}
 by definition.  
By heuristic computation we see that those   are  nilpotent as in \eqref{eq:Qnil-intro}:
\begin{equation}
\label{eq:Qnicnil}
\Qnic^{2}=0=\Qnicastsq.
\end{equation}
The  Hamiltonian is defined by the  supersymmetric form   \eqref{eq:Hsusy}:
\begin{equation}
\label{eq:Hnicdef}
\Hnic:=\{\Qnic,\; \Qnicast \}.
\end{equation}
By direct computation we see  that 
\begin{align}
\label{eq:Hnicgutai}
\Hnic&=\sum_{i\in \Z}\bigl\{
a^{\ast}_{2i}a_{2i-1}a_{2i+2}a^{\ast}_{2i+3}
+a^{\ast}_{2i-1}a_{2i}a_{2i+3}a^{\ast}_{2i+2} \nonumber\\
&\ \ \ +a^{\ast}_{2i}a_{2i}a_{2i+1}a^{\ast}_{2i+1}
+a^{\ast}_{2i-1}a_{2i-1}a_{2i}a^{\ast}_{2i}
- a^{\ast}_{2i-1} a_{2i-1}a_{2i+1}a^{\ast}_{2i+1} \bigr\}.
\end{align}
The fermion lattice model defined  by  the above Hamiltonian $\Hnic$ 
 which is generated by the pair of  supercharges $\Qnic$ and 
$\Qnicast$ satisfies   ${\cal{N}}=2$  supersymmetry.
  Hereafter we call this supersymmetry  model  {\it{the Nicolai model}}.

For later sake we shall decompose   the Hamiltonian 
into the classical term $H_{\rm{classical}}$ and the hopping term
 $H_{\rm{hop}}$ as 
 \begin{align}
\label{eq:Hnic-decom}
\Hnic=
H_{\rm{classical}}
+H_{\rm{hop}}
\end{align}
by setting 
\begin{align}
\label{eq:Hnic-cla}
H_{\rm{classical}}
:=\sum_{i\in \Z}
n_{2i}-n_{2i-1}n_{2i}
-n_{2i}n_{2i+1} +n_{2i-1}n_{2i+1},
\end{align}
and 
\begin{align}
\label{eq:Hnic-hop}
H_{\rm{hop}}:=\sum_{i\in \Z}
a^{\ast}_{2i}a_{2i-1}a_{2i+2}a^{\ast}_{2i+3}
+a^{\ast}_{2i-1}a_{2i}a_{2i+3}a^{\ast}_{2i+2}. 
\end{align}

We can immediately see 
\begin{align}
\label{eq:HnicUone}
\gamma_{\theta} (\Hnic)=\Hnic\  \ \forall \theta \in [0,\; 2\pi),
\end{align}
and 
\begin{align}
\label{eq:sigk}
\sigma_{k} (\Hnic)=\Hnic
\quad \forall k  \in 2\Z.
\end{align}
Therefore  the Nicola model  has global 
$U(1)$-symmetry and  $\Z_2$-translation  symmetry  in space (periodicity).
From   
\begin{equation}
\label{eq:Qnicrho}
\rho(\Qnic)=-\Qnicast,\quad 
\rho(\Qnicast)=-\Qnic,
\end{equation}
 and \eqref{eq:Hnicdef}   the model  has the particle-hole symmetry 
\begin{align}
\label{eq:Hnicrhoinv}
\rho(\Hnic)=\Hnic.
\end{align}

We shall reformulate the Nicolai  model 
 as a supersymmetric $\cstar$-dynamical system based on  \cite{AHP16}.
 In this  framework we use  superderivations
 as our  basic building block. 
By  the formal  supercharge operators  
\eqref{eq:Qnic} \eqref{eq:Qnicast} 
   the  following pair of  superderivations are defined rigorously
\begin{equation}
\label{eq:delchanic-IMP}
\delchanic(A):=  
[\Qnic, \;  A]_{\gamma} {\text {\ \ for every \ }} A\in \core,
 \end{equation}
and 
\begin{equation}
\label{eq:delchanicast-IMP}
\delchanicast(A):= 
[\Qnicast, \;  A]_{\gamma} {\text {\ \ for every \ }} A\in \core.
 \end{equation}
Also  we  define 
the  time generator by  
\begin{equation}
\label{eq:Dercha-IMP}
\Derchanic(A):=[\Hnic, \;  A] {\text {\ \ for every \ }} A\in \core.
 \end{equation}

By definition 
the linear map $\delchanic$ 
   satisfies 
\begin{equation}
\label{eq:oddsuperder}
\delchanic\cdot \gamma=-\gamma \cdot \delchanic\ \ 
{\text {on}}\ \core, 
\end{equation}
and  also   the {\it{graded Leibniz rule}}:
\begin{equation}
\label{eq:gleib}
\delchanic(AB)=
\delchanic(A)B+\gamma (A) \delchanic(B)\ \  {\text {for every}}\ A, B\in \core.
\end{equation}
By \eqref{eq:Qnicnil}
 the nilpotent condition is  satisfied: 
\begin{equation}
\label{eq:nil-delcha}
\delchanic\cdot\delchanic=\zeromap =\delchanicast\cdot\delchanicast 
{\text {\ \ on\ \ }} \core. 
\end{equation}
From  \eqref{eq:Hnicdef} 
 the identity that expresses  supersymmetry  follows: 
\begin{equation}
\label{eq:derivationSUSY}
\Derchanic= 
\delchanicast\cdot\delchanic+ \delchanic\cdot\delchanicast 
 {\text {\ \ on  \ }} \core.
 \end{equation}

A state $\vp$ on $\Al$ is called supersymmetric if it is 
invariant under $\delchanic$:
\begin{equation}
\label{eq:susyinv}
\vp \left(\delchanic(A)\right)=0\ \ {\text{for every}}\  A\in \core.
\end{equation}
We will see that there are lots  of  supersymmetric states 
 for the Nicolai model.

 In the remainder of this section, we shall collect crucial results on supersymmetric dynamics.
In fact, those are valid  for general  supersymmetric fermion lattice models 
 of finite-range interactions \cite{AHP16}.
\begin{proposition}
\label{prop:AHPmain}
 There exists  a strongly continuous one parameter
group of $\ast$-automorphisms $\alchanict$ ($t\in \R$)
on  $\Al$ whose pre-generator  is given by 
the derivation $\Derchanic\equiv 
\delchanicast\cdot\delchanic+ \delchanic\cdot\delchanicast$ defined on the local algebra 
 $\core$. 
\end{proposition}

From  Proposition  \ref{prop:AHPmain} we can derive   the 
 key  statement in this paper. 
\begin{proposition}
\label{prop:GENtimeinvariant}
 Suppose that   $B\in \core$ is  
annihilated by both   superderivations  
  $\delchanic$ and $\delchanicast${\rm{:}}
\begin{equation}
\label{eq:Bsuper-const}
\delchanic(B)=0=\delchanicast(B).
 \end{equation}
Then it  is invariant under 
the time evolution{\rm{:}}
 \begin{equation}
\label{eq:Btime-const}
\alchanict(B)=B  \quad {\text{for all}}\ t\in \R.
 \end{equation}
\end{proposition}

 \begin{proof}
From the identity  \eqref{eq:derivationSUSY} 
 the assumption \eqref{eq:Bsuper-const} yields 
$\Derchanic(B)=0$.
As   $\Derchanic$ is a pre-generator for 
 the  strongly continuous one parameter
group of $\ast$-automorphism $\alchanict$ ($t\in \R$)
on  $\Al$ by  Proposition \ref{prop:AHPmain},  
 the equation  \eqref{eq:Btime-const} follows. 
\end{proof}

The following statement given in \cite{AHP16}  
 is also useful.
\begin{proposition}
\label{prop:SUSYground}
If a state  $\vp$  on $\Al$ is a supersymmetric state for $\delchanic$, 
 then $\vp$ is  a ground state
 (in the sense of Definition 5.3.18 of \cite{BRII})
 for  the one-parameter   
group of $\ast$-automorphisms $\alchanict$ ($t\in \R$).
 In particular, $\vp$ is invariant under $\alchanict$ ($t\in \R$).
\end{proposition}

\section{Local fermionic  constants  of  motion}
\label{sec:INFINITE-CONSTANT}
The purpose of this section is to systematically provide  infinitely many   local fermionic constants  of motion for 
the Nicolai  model.
\subsection{Classical sequences   that encode local fermionic  
constants  of  motion}
\label{subsec:LOC-forb-permit-seq}
In this subsection 
 we introduce the sequences 
 that encode local fermionic constants of motion.
\begin{definition} 
\label{defn:forbid-permit}
Let  $\I$ denote  an  interval  of $\Z$, i.e. 
$\I=[m, n]=\{m, m+1, \cdots, n-1, n\}$ with $m,n\in \Z$ ($m<n$).
Let $f$ be a $\{-1, +1\}$-valued function  on  $\I$.
If either     
\begin{align}
\label{eq:2i-forbid}
f(2i-1)=-1,\ \ f(2i)=+1,\ \  f(2i+1)=-1,  
\end{align}
or 
\begin{align}
\label{eq:2i-forbid-ura}
f(2i-1)=+1,\ \ f(2i)=-1,\ \ f(2i+1)=+1,
\end{align}
 holds  for some  
$\{2i-1,\; 2i,\; 2i+1\}\subset \I$  ($i\in\Z$),  
then  $f$ is said to be  forbidden.
Otherwise, if $f$ does not include such  subsequences 
  \eqref{eq:2i-forbid} \eqref{eq:2i-forbid-ura} anywhere  in  $\I$,
 then it is said to be permitted.
 The  set of all $\{-1, +1\}$-valued permitted sequences   on  $\I$ 
 is  denoted by  ${\Xi}_{\I}$.
\end{definition}

We will frequently use 
 the  intervals whose  edges are both  even:
\begin{align} 
\label{eq:Ikl}
\Ikl\equiv[2k, 2l],
 \quad    k,l\in\Z\quad  \text{such that}\quad k<l.
\end{align}
By definition  $\Ikl$ 
 has  $2(l-k)+1(\ge 3)$ sites.
 We  intend to  find  fermion operators  on  $\Ikl$ that are invariant 
 under the time evolution of  the infinite system $\Al$.
For this sake  we introduce  a subclass of permitted sequences  on $\Ikl$ 
 imposing some additional requirement upon the  edges.

\begin{definition}
\label{defn:seq-cons}
Let $f$ be a $\{-1, +1\}$-valued permitted sequence  on the  interval  $\Ikl$
  with  $k,l\in\Z$  $(k<l)$ of  \eqref{eq:Ikl}, 
namely  $f\in \Xi_{\Ikl}$ as in  Definition \ref{defn:forbid-permit}.
Assume  that   
$f$   takes  a constant  on the  left-end  pair sites   $\{2k,\; 2k+1\}$,  
and  that  $f$  takes  a constant  
on the right-end  pair sites $\{2l-1,\; 2l\}$. Namely  
\begin{align}
\label{eq:leftedge-seq}
f(2k)=f(2k+1)=+1\quad  \text{or}\quad 
f(2k)=f(2k+1)=-1
\end{align}
and 
\begin{align}
\label{eq:rightedge-seq}
f(2l-1)=f(2l)=+1\quad  \text{or}\quad 
f(2l-1)=f(2l)=-1.
\end{align}
The set of all  
$\{-1, +1\}$-valued permitted sequences
on $\Ikl$ satisfying the above  marginal conditions 
 on both edges is  denoted by  $\hat{\Xi}_{k,l}$.
The union of   $\hat{\Xi}_{k,l}$ over all  
$k,l\in\Z$ ($k<l$) is  denoted by $\hat{\Xi}$.
Each $f\in\hat{\Xi}$  is called a   local  sequence of conservation  for the Nicolai model.
\end{definition}

\begin{remark} 
\label{rem:marginal}
The  requirements 
\eqref{eq:leftedge-seq}
\eqref{eq:rightedge-seq} on the edges of $\Ikl$  are   essential to 
 make  conservations  for the Nicolai model on the infinite lattice $\Z$.
 \end{remark}

\begin{remark} 
\label{rem:crude}
By  crude estimate we  see that  the number of local 
sequences of  conservation in  
$\hat{\Xi}_{k,l}$ is roughly $(\frac{2^{3}-2}{2})^{(l-k)}=3^{(l-k)}=3^{m/2}$, 
 where  $m=2(l-k)$  denotes approximately the  size of the system 
(i.e.  the number of sites in  $\Ikl$).
\end{remark}

It is convenient to  consider   the following  special elements 
 of $\hat{\Xi}$. 
\begin{definition}
\label{defn:ID-conserved}
For each 
 $k,l\in\Z$ ($k<l$) let 
 $\IDplkl\in \hat{\Xi}_{k,l}$ and
 $\IDmikl\in \hat{\Xi}_{k,l}$ 
denote  the constants 
over   $\Ikl$ taking   $+1$
 and $-1$, respectively:
\begin{equation}
\label{eq:constkl}
\IDplkl(i)=+1 \  \forall i\in \Ikl,\quad    
\IDmikl(i)=-1 \  \forall i\in \Ikl. 
\end{equation}
\end{definition}

\subsection{Construction of local fermionic  constants of motion}
\label{subsec:LCS}
We shall  give  a rule to assign a 
  local fermion operator  for every    
local  sequence   of conservation in      
  $\hat{\Xi}$ of   Definition \ref{defn:seq-cons}.
\begin{definition}
\label{defn:local-assignment}
For each $i\in\Z$
  let $\zetai$ denote the assignment   
from $\{-1,+1\}$ into  the fermion annihilation-creation operators  at $i$ as  
\begin{align}
\label{eq:zeta}
\zetai(-1):= \ai,\quad \zetai(+1):= \aicr.
\end{align}
 Take  any  pair of integers $k,l\in\Z$ such that $k<l$.
For  each  $f\in \hat{\Xi}_{k,l}$, set 
\begin{align}
\label{eq:Qseq}
\Qseq(f)&:=\prod_{i=2k}^{2l} \zetai\left(f(i)\right)\nonumber\\
&\equiv \zeta_{2k} 
\left(f(2k)\right) \zeta_{2k+1}\left(f(2k+1)\right) 
\cdots 
\cdots \zeta_{2l-1}\left(f(2l-1)\right) 
\zeta_{2l}\left(f(2l)\right)\in 
{\Al}(\Ikl)_{-},
\end{align}
 where the  multiplication 
  is  taken in the  increasing  order as above. 
The  formulas \eqref{eq:Qseq} for  all   $k,l\in\Z$ ($k<l$)
   yield   a unique   assignment $\Qseq$  from $\hat{\Xi}$ into $\coreo$.  
\end{definition}

By   Definition \ref{defn:local-assignment},
  for $k,l\in\Z$ ($k<l$)
\begin{align}
\label{eq:Qseqconst}
\Qseq(\IDplkl)&:=\atwokcr \atwokpcr \cdots 
\atwolmcr \atwolcr \in 
{\Al}(\Ikl)_{-}, \nonumber\\
\Qseq(\IDmikl)&:=\atwok \atwokp \cdots 
\atwolm \atwol \in 
{\Al}(\Ikl)_{-}.
\end{align}

\bigskip
\bigskip
\noindent{\bf{Examples}}\\
We will give concrete examples  for 
 local  sequences of conservation 
   of   Definition \ref{defn:seq-cons} and 
  their associated   local fermion operators 
   of   Definition \ref{defn:local-assignment}.
First we  see that  $\hat{\Xi}_{0,1}$  on $\Izeroone\equiv[0,1,2]$
  consists of  two obvious one:

\begin{center}
\begin{tabular}{|c||ccc|}
\hline
${\hat{\Xi}_{0,1}}$ &$0$ &$1$ &$2$  \\ \hline
$\IDmizeroone$ &$-1$ &$-1$ &$-1$   \\
\hline
$\IDplzeroone$ &$+1$ &$+1$ &$+1$  \\
 \hline\end{tabular}
\end{center}
By  \eqref{eq:Qseq}  of Definition \ref{defn:local-assignment}
 the corresponding local fermion operators are  
\begin{align}
\label{eq:Xi0-1}
\Qseq(\IDmizeroone)&=a_{0} a_{1} a_{2}
\in {\Al}(\Izeroone)_{-},\nonumber\\
\Qseq(\IDplzeroone)&=a^{\ast}_{0} a^{\ast}_{1} a^{\ast}_{2}
\in {\Al}(\Izeroone)_{-}.
\end{align}

Next we consider the segment  $\Izerotwo\equiv[0,1,2,3,4]$ 
  by setting  $k=0$ and $l=2$. 
The space 
$\hat{\Xi}_{0,2}$ on $\Izerotwo$
 consists of the following five  sequences of conservation:
\begin{center}
\begin{tabular}{|c||cc|c|cc|}
\hline
${\hat{\Xi}_{0,2}}$ &$0$ &$1$ &$2$ &$3$ &$4$ \\ \hline
$\IDmizerotwo$ &$-1$ &$-1$ &$-1$ &$-1$ &$-1$  \\
\hline
$u_{[0,4]}^{\rm{i}}$  &$-1$ &$-1$ &$-1$ &$+1$ &$+1$  \\
$u_{[0,4]}^{\rm{ii}}$  &$-1$ &$-1$ &$+1$ &$+1$ &$+1$  \\
\hline
$v_{[0,4]}^{\rm{i}}$  &$+1$ &$+1$ &$+1$ &$-1$ &$-1$  \\
$v_{[0,4]}^{\rm{ii}}$  &$+1$ &$+1$ &$-1$ &$-1$ &$-1$  \\
\hline
$\IDplzerotwo$ &$+1$ &$+1$ &$+1$ &$+1$ &$+1$  \\
 \hline\end{tabular}
\end{center}
Note that 
\begin{align}
\label{eq:Xi0-2-minus}
\IDmizerotwo=-\IDplzerotwo,\ u_{[0,4]}^{\rm{i}}=-v_{[0,4]}^{\rm{i}},\ 
u_{[0,4]}^{\rm{ii}}=-v_{[0,4]}^{\rm{ii}}.
\end{align}
By \eqref{eq:Qseq}  of Definition \ref{defn:local-assignment}
we have 
\begin{align}
\label{eq:Xi0-2}
\Qseq(\IDmizerotwo)&=a_{0} a_{1} a_{2}a_{3} a_{4}
\in {\Al}(\Izerotwo)_{-},\nonumber\\
\Qseq(u_{[0,4]}^{\rm{i}})&=a_{0} a_{1} a_{2}a^{\ast}_{3} a^{\ast}_{4}\in {\Al}
(\Izerotwo)_{-},\nonumber\\
\Qseq(u_{[0,4]}^{\rm{ii}})&=a_{0} a_{1} a^{\ast}_{2}a^{\ast}_{3} a^{\ast}_{4}
\in {\Al}(\Izerotwo)_{-},\nonumber\\
\Qseq(v_{[0,4]}^{\rm{i}})&=a^{\ast}_{0} a^{\ast}_{1} a^{\ast}_{2}a_{3} a_{4}
\in {\Al}(\Izerotwo)_{-},\nonumber\\
\Qseq(v_{[0,4]}^{\rm{ii}})&=a^{\ast}_{0} a^{\ast}_{1} a_{2}a_{3}a_{4}
\in {\Al}(\Izerotwo)_{-},\nonumber\\
\Qseq(\IDplzerotwo)&=a^{\ast}_{0} a^{\ast}_{1} a^{\ast}_{2}a^{\ast}_{3} a^{\ast}_{4}
\in {\Al}(\Izerotwo)_{-}.
\end{align}

We then consider    the segment  $\Izerothree\equiv[0,1,2,3,4,5,6]$ 
  taking   $k=0$ and $l=3$. 
By definition it consists  of  $5+4+4+5=18$  sequences of conservation:
\begin{center}
\begin{tabular}{|c||cc|ccc|cc|}
\hline
${\hat{\Xi}_{0,3}}$ &$0$ &$1$ &$2$ &$3$ &$4$ &$5$ &$6$ \\ \hline
$s_{[0,6]}^{\circ} $ &$-1$ &$-1$ &$-1$ &$-1$ &$-1$ &$-1$ &$-1$  \\
$s_{[0,6]}^{\rm{i}}$ &$-1$ &$-1$ &$-1$ &$+1$ &$+1$ &$-1$ &$-1$  \\
$s_{[0,6]}^{\rm{ii}}$ &$-1$ &$-1$ &$+1$ &$+1$ &$-1$ &$-1$ &$-1$  \\
$s_{[0,6]}^{\rm{iii}}$ &$-1$ &$-1$ &$-1$ &$+1$ &$-1$ &$-1$ &$-1$  \\
$s_{[0,6]}^{\rm{iv}}$ &$-1$ &$-1$ &$+1$ &$+1$ &$+1$ &$-1$ &$-1$  \\
\hline
$u_{[0,6]}^{\rm{i}}$  &$-1$ &$-1$ &$-1$ &$-1$ &$-1$ &$+1$ &$+1$  \\
$u_{[0,6]}^{\rm{ii}}$   &$-1$ &$-1$ &$-1$ &$-1$ &$+1$ &$+1$ &$+1$   \\
$u_{[0,6]}^{\rm{iii}}$   &$-1$ &$-1$ &$-1$ &$+1$ &$+1$ &$+1$ &$+1$   \\
$u_{[0,6]}^{\rm{iv}}$   &$-1$ &$-1$ &$+1$ &$+1$ &$+1$ &$+1$ &$+1$   \\
\hline
$v_{[0,6]}^{\rm{i}}$   &$+1$ &$+1$ &$+1$ &$+1$ &$+1$ &$-1$ &$-1$   \\
$v_{[0,6]}^{\rm{ii}}$   &$+1$ &$+1$ &$+1$ &$+1$ &$-1$ &$-1$ &$-1$   \\
$v_{[0,6]}^{\rm{iii}}$   &$+1$ &$+1$ &$+1$ &$-1$ &$-1$ &$-1$ &$-1$   \\
$v_{[0,6]}^{\rm{iv}}$   &$+1$ &$+1$ &$-1$ &$-1$ &$-1$ &$-1$ &$-1$  \\
\hline
$t_{[0,6]}^{\bullet}$ &$+1$ &$+1$ &$+1$ &$+1$ &$+1$ &$+1$ &$+1$  \\
$t_{[0,6]}^{\rm{i}}$ &$+1$ &$+1$ &$+1$ &$-1$ &$-1$ &$+1$ &$+1$  \\
$t_{[0,6]}^{\rm{ii}}$ &$+1$ &$+1$ &$-1$ &$-1$ &$+1$ &$+1$ &$+1$  \\
$t_{[0,6]}^{\rm{iii}}$ &$+1$ &$+1$ &$+1$ &$-1$ &$+1$ &$+1$ &$+1$  \\
$t_{[0,6]}^{\rm{iv}}$ &$+1$ &$+1$ &$-1$ &$-1$ &$-1$ &$+1$ &$+1$  \\
 \hline\end{tabular}
\end{center}
Note that 
$s_{[0,6]}^{\circ}\equiv \IDmizerothree$ and 
$t_{[0,6]}^{\bullet}\equiv \IDplzerothree$
and that 
\begin{align}
\label{eq:Xi0-3-minus}
s_{[0,6]}^{\circ}=-t_{[0,6]}^{\bullet},\ 
s_{[0,6]}^{k}=-t_{[0,6]}^{k},\forall k\in \{\rm{i}, \rm{ii}, \rm{iii}, \rm{iv}\} \nonumber\\
u_{[0,6]}^{k}=-v_{[0,6]}^{k}, 
\forall k\in \{\rm{i}, \rm{ii}, \rm{iii}, \rm{iv}\}.
\end{align}
According to the rule 
we have  18 fermion operators associated to
 ${\hat{\Xi}_{0,3}}$.

\bigskip
\bigskip

We note some  properties  of the local fermion operators
  of  Definition \ref{defn:local-assignment}.
\begin{lemma}
\label{lem:Qseq-ast}
For every  $f\in\hat{\Xi}_{k,l}$ with  $k,l\in\Z$ ($k<l$),
 $-f$ also belongs to  $\hat{\Xi}_{k,l}$.
For each    
 $f\in\hat{\Xi}_{k,l}$
\begin{align}
\label{eq:Qseqast}
{\Qseq(f)}^{\ast}=(-1)^{\sigma(k,l)}\Qseq(-f),
\end{align}
where 
\begin{align}
\label{eq:skl}
\sigma(k,l)\equiv \left\{2(l-k)+1 \right\}(l-k).
\end{align}
In particular, 
\begin{align}
\label{eq:Qseqconstast}
{\Qseq(\IDmikl)}^{\ast}=(-1)^{\sigma(k,l)}\Qseq(\IDplkl).
\end{align}
\end{lemma}
\begin{proof}
Obvious.
\end{proof}

\begin{proposition}
\label{prop:ALGQseq}
For every  
 $f\in\hat{\Xi}$,  the local fermion operator $\Qseq(f)\in \coreo$ 
and its adjoint $\Qseq(f)^{\ast}\in\coreo$ are nilpotent{\rm{:}}
\begin{align}
\label{eq:Qseq-nil}
\Qseq(f)^{2}=0=\Qseq(f)^{\ast\,2}.
\end{align}
For each 
 $f, g\in\hat{\Xi}$ 
\begin{equation}
\label{eq:anticommute}
\{\Qseq(f), \;\Qseq(g) \}=0
\end{equation}
is satisfied unless 
the support of $f$ and the support   of $g$ have 
a non-empty intersection  on which 
   $f=-g$ holds. 
\end{proposition}
\begin{proof}
The statements    can be  verified  by noting 
 the   form of $f\in \hat{\Xi}$ of    
 Definition \ref{defn:seq-cons} and 
  the  form of   $\Qseq(f)$  of  
Definition \ref{defn:local-assignment} 
together with  some  obvious  identities of fermion operators:  $\ai \ai =\aicr \aicr=0$,  
$\aicr\ai=n_i $ and $\ai \aicr=1-n_i$   
for   $i\in\Z$.
\end{proof}

 Let us   define   algebras  
 generated  by these  local fermion operators.
\begin{definition}
\label{defn:CA}
 Let $\CA$ denote the $\ast$-subalgebra in $\core$
  finitely  generated by 
 $\{\Qseq(f)\in \core|\; f\in \hat{\Xi} \}$.
 For every    $k,l\in\Z$ such that $k<l$
the $\ast$-subalgebra  
generated by 
 $\{\Qseq(f)\in \core|\; 
f\in \hat{\Xi}_{k^{\prime},\; l^{\prime}},  k\le k^{\prime} 
<l^{\prime}\le  l  \}$
 is  denoted by  $\CA({k,l})$.
 (By definition   $\CA({k,l})\subset {\Al}(\Ikl)$, 
 and $\CA({k,l})\supset \CA({p,q})$ if  
  $k\le p < q \le  l$.)
 \end{definition}
We are in a position to state our main result with the above definition.
\begin{theorem}
\label{thm:NICconstant}
Let $\alchanict$ ($t\in \R$) denote the 
time evolution for the Nicolai model  given  in $\S$\ref{subsec:NIC}.
Then for every  $B\in\CA$
 \begin{equation}
\label{eq:Btime-const-NIC}
\alchanict(B)=B  \quad {\text{for all}}\ t\in \R.
 \end{equation}
In particular,  for every 
 $f\in\hat{\Xi}$,
 \begin{equation}
\label{eq:Btime-const-NICparticular}
\alchanict(\Qseq(f))=\Qseq(f),\quad
\alchanict(\Qseq(f)^{\ast})=\Qseq(f)^{\ast} 
  \quad {\text{for all}}\ t\in \R.
 \end{equation}
\end{theorem}

\begin{proof}
By  Proposition \ref{prop:GENtimeinvariant}
 it suffices to show that 
  \begin{equation}
\label{eq:nic-inv}
\delchanic(B)=0, \quad \delchanicast(B)=0
\ \ \text{for every}\ B\in \CA.
 \end{equation}
Furthermore, by 
   Definition \ref{defn:CA} and 
the graded Leibniz rule of superderivations \eqref{eq:gleib}
 we need   to  show that 
 \begin{equation}
\label{eq:nic-inv-del}
\delchanic\left(\Qseq(f)\right)=0, \quad \delchanicast\left(\Qseq(f)\right)=0
\ \ \text{for every}\  f\in \hat{\Xi}.
\end{equation}
From 
 Definitions \ref{defn:forbid-permit}, \ref{defn:seq-cons} and  \ref{defn:local-assignment}, 
  by  using the CARs,  
  we see that for  all   $i \in \Z$
\begin{align}
\label{eq:vanishQseq}
\Qseq(f)\chanic(\{2i-1, 2i, 2i+1\})=0=\chanic(\{2i-1, 2i, 2i+1\}) \Qseq(f),\nonumber\\
 \Qseq(f)\chanicast(\{2i-1, 2i, 2i+1\})=0=\chanicast(\{2i-1, 2i, 2i+1\})\Qseq(f).
\end{align}
These yield \eqref{eq:nic-inv-del}.
\end{proof}

 Theorem \ref{thm:NICconstant} can be expressed  in a more heuristic manner:
\begin{equation}
\label{eq:easyQ}
\{\Qnic,\; \Qseq(f) \}=0=\{\Qnicast \;  \Qseq(f)\} 
\ \ \text{for every}\ f\in\hat{\Xi},
\end{equation}
and 
\begin{equation}
\label{eq:easyQfconst}
[\Hnic, \;  \Qseq(f)]=0 \ \ \text{for every}\ f\in\hat{\Xi}.
\end{equation}

We  shall  provide  some  terminologies relevant to Theorem \ref{thm:NICconstant}.
\begin{definition}
\label{defn:const-motion}
For each local  sequence of conservation  
 $f\in\hat{\Xi}$,  
$\Qseq(f)$ is called the  local fermionic  constant of 
   motion  associated  to $f$.
 The pair of nilpotent local fermion operators $\{\Qseq(f), \Qseq(f)^{\ast}\}$ 
 is called the hidden  fermion charge associated to $f$.
 The $\ast$-algebra  $\CA$  in   Definition \ref{defn:CA}
is called the algebra of the constants of motion for the Nicolai model, 
and  $\CA({k,l})$  with  $k,l\in\Z$ ($k<l$)  is 
called the algebra of the constants of motion 
  within the segment $\Ikl$.
\end{definition}

\begin{remark} 
\label{rem:gauge}
The subalgebras $\CA$ and 
  $\CA({k,l})$ for any  $k,l\in\Z$ ($k<l$)
 include   many   observables (self-adjoint operators).
 Some of them  are $U(1)$-gauge invariant observables.
 \end{remark}

\begin{remark} 
 Padmanabhan et al. studied  non-ergodic dynamics of supersymmetric fermion lattice models in  \cite{PADMAN}.
 Our   method and the   
local constants of motion given  here   
  are different from those in  \cite{PADMAN}.
 \end{remark}

\section{Highly degenerated classical supersymmetric ground  states}
\label{sec:C-G-States}
In this section we  shall provide  all classical supersymmetric ground   
states of  the Nicolai  model. 
 We will not discuss  general  supersymmetric ground 
 states.

\subsection{Classical configurations}
\label{subsec:classical}
 Let   
$\ket{0}_{i}$ and  $\ket{1}_{i}$  
 denote  the  empty-state vector and 
the occupied-state vector 
 of the  spinless fermion  at  $i\in \Z$, respectively.
Thus 
\begin{equation}
\label{eq:keti}
\ai \ket{1}_{i}= \ket{0}_{i},\ 
\aicr \ket{1}_{i}=0,\ 
\aicr \ket{0}_{i}= \ket{1}_{i},\ 
\ai \ket{0}_{i}=0.
\end{equation}
With  $\{\ket{0}_{i},\ \ket{1}_{i};\; i\in \Z\}$
 we generate the Fock space. 

\begin{definition}
\label{defn:CLASSIC}
Let $g(n)$ denote an arbitrary   
$\{0, 1\}$-valued  function  over  $\Z$. 
It is called  a classical configuration over $\Z$.
For any classical configuration $g(n)$ define 
\begin{align}
\label{eq:gn-vector}
\ket{g(n)_{n\in\Z}}:= \cdots \otimes 
\ket{g(i-1)}_{i-1} \otimes  \ket{g(i)}_{i}
\otimes \ket{g(i+1)}_{i+1} \otimes   \cdots
\end{align}
This infinite product vector determines  
a  state $\psi_{g(n)}$ on the  fermion system  $\Al$. 
 It  will be  called  the  classical state   associated 
to the configuration  $g(n)$ over $\Z$.
\end{definition}

 To each  classical configuration over $\Z$
 we   assign  an operator by the following rule. 
\begin{definition}
\label{defn:OPERA}
For each $i\in\Z$
  let $\opkappai$ denote the map 
from $\{0, 1\}$ into  $\Ali$ given as  
\begin{align}
\label{eq:OPERAi-config}
\opkappai(0):= \id, \quad \opkappai(1):= \aicr.
\end{align}
For each  classical configuration  $g(n)$ over $\Z$ 
 define   
the   infinite-product of fermion field operators: 
\begin{align}
\label{eq:defOPERAg}
\OPERA(g)&:=\prod_{i\in\Z} \opkappai\left(g(i)\right)=
\cdots \opkappaiminus\left(g(i-1)\right) \opkappai\left(g(i)\right) 
\opkappaiplus\left(g(i+1)\right)
\cdots, 
\end{align}
where the  multiplication  is  taken  in the  increasing  order. 
If  $g(n)$ has  a  compact support, then   
\begin{align}
\label{eq:OPERAgcompact}
\OPERA(g)\in \core.
\end{align}
Otherwise  $\OPERA(g)$
 denotes  a  formal operator which is  out of  $\Al$. 
\end{definition}

One can naturally  relate    Definition \ref{defn:CLASSIC} (product vectors)
 and   Definition \ref{defn:OPERA} (product  operators) 
via   the Fock representation.
\begin{proposition}
\label{prop:corresonFock}
Let  $\eta_0$ denote the  Fock  vector (no-particle wave function).
For any classical configuration  $g(n)$ over $\Z$, 
 the following  identity   holds{\rm{:}} 
\begin{align}
\label{eq:gonFock}
\ket{g(n)_{n\in\Z}}=\OPERA(g)\eta_0.
\end{align}
\end{proposition}

\begin{proof}
This  directly  follows from  Definition \ref{defn:CLASSIC}
 and Definition \ref{defn:OPERA} by noting \eqref{eq:keti}.
\end{proof}

\begin{remark} 
\label{rem:noncompactg}
Even when 
 $g(n)$  does not have a compact support, 
 the identity  \eqref{eq:gonFock}
 of Proposition \ref{prop:corresonFock} is valid.
 For example,  take  the constant 
  $\iota(n):=1$  $\forall n\in \Z$ for the  classical 
configuration over $\Z$.
Obviously   the  support  of $\iota$ is non compact. 
Nevertheless,  we have 
\begin{align*}
&\ \ \OPERA(\iota)\eta_0 \nonumber\\
&=\cdots a^{\ast}_{-2}
 \ket{0}_{-2}
\otimes a^{\ast}_{-1} \ket{0}_{-1} \otimes a^{\ast}_{0} \ket{0}_{0}
\otimes a^{\ast}_{1} \ket{0}_{1} \otimes a^{\ast}_{2} \ket{0}_{2}
\otimes   \cdots\nonumber\\
&=\cdots 
\otimes \ket{1}_{-2}
\otimes \ket{1}_{-1} \otimes  \ket{1}_{0}
\otimes \ket{1}_{1} \otimes  \ket{1}_{2}
\otimes   \cdots\equiv \eta_1,
\end{align*}
where $\eta_1$ denotes  
the fully occupied wave function over $\Z$.
\end{remark}

\subsection{Classical supersymmetric ground states}
\label{subsec:classical}
We introduce  the following special class of classical configurations.
 \begin{definition}
\label{defn:ground-config}
 Take  any  three-site  subset  $\{2i-1, 2i, 2i+1\}$ centered 
 at an even site $2i$ ($i\in\Z$).
   Among $8$ configurations  ($\{0,1\}$-valued functions)
on  $\{2i-1, 2i, 2i+1\}$,  
   $``0, 1, 0"$ and  $``1, 0, 1"$  are   called   forbidden triplets.
If  a classical configuration $g(n)$ $(n\in\Z)$ 
 does not include such  forbidden triplets  over $\Z$, 
then it is called a ground-state  configuration for  the Nicolai model 
 over $\Z$.
The set of all ground-state  configurations for  the Nicolai model over $\Z$
 is  denoted by  $\Upsilon$.
The set of all ground-state  configurations for  the Nicolai model
 whose support is included in some finite region 
 is denoted by  $\Upsilon_{\circ}$.
The set of all ground-state  configurations for  the Nicolai model
 whose support is included in 
a finite region $\I\Subset\Z$ is denoted by   $\Upsilon_{\I}$.
\end{definition}

We can  classify all  classical supersymmetric ground states 
 by using  Definition  \ref{defn:ground-config}.
\begin{theorem}
\label{thm:Allclassic}
A classical state  on the fermion lattice system  $\Al$
  is a supersymmetric ground state of  the Nicolai model if and only if 
its associated configuration  
 $g(n)$ over $\Z$ is a ground-state 
 configuration  as  in  Definition
   \ref{defn:ground-config} (i.e. $g(n)\in \Upsilon$).
Every  such  state 
 is  invariant under  
 the time evolution $\alchanict$ ($t\in \R$).  
\end{theorem}

\begin{proof}
First we shall see  the  action of the local fermion charges
 $\chanictwoisan\equiv a_{2i+1} a^{\ast}_{2i} a_{2i-1}$  ($i\in \Z$) defined  in  \eqref{eq:Qnic} 
  upon classical states 
(via  the GNS representation   for the Fock state).
If  the classical configuration $g(n)$ over $\Z$ satisfies 
 $g(2i-1)=1, g(2i)=0 ,g(2i+1)=1$, 
namely there  includes 
the  forbidden  $``1, 0, 1"$ on  $\{2i-1, 2i, 2i+1\}$,
 then 
\begin{align}
\label{eq:chanic-act-nonzero}
&\chanictwoisan
 \ket{g(n)_{n\in\Z}} \nonumber\\ 
&=\cdots \otimes \ket{g(2i-3)}_{2i-3} \otimes  \ket{g(2i-2)}_{2i-2} 
 \nonumber\\
&\ \otimes \ket{0}_{2i-1} \otimes  \ket{1}_{2i}
\otimes \ket{0}_{2i+1} \otimes   
\ket{g(2i+2)}_{2i+2} \otimes 
\cdots,
\end{align}
where any entry 
 on  the complement of  $\{2i-1, 2i, 2i+1\}$ in $\Z$ is  unchanged.
For any other $g(n)$  the  corresponding classical  
vector $\ket{g(n)_{n\in\Z}}$ is always deleted by 
 $\chanictwoisan$:
\begin{align}
\label{eq:chanic-act-zero}
\chanictwoisan
\ket{g(n)_{n\in\Z}}  
=0.
\end{align}
Similarly, consider the action of $\chanictwoisanast$.
If  the classical configuration $g(n)$ over $\Z$ satisfies 
 $g(2i-1)=0, g(2i)=1 ,g(2i+1)=0$, namely there  includes 
the  forbidden  $``0, 1, 0"$ on  $\{2i-1, 2i, 2i+1\}$, then 
\begin{align}
\label{eq:chanicast-act-nonzero}
&\chanictwoisanast
 \ket{g(n)_{n\in\Z}} \nonumber\\ 
&=\cdots \otimes \ket{g(2i-3)}_{2i-3} \otimes  \ket{g(2i-2)}_{2i-2} 
 \nonumber\\
&\ \otimes \ket{1}_{2i-1} \otimes  \ket{0}_{2i}
\otimes \ket{1}_{2i+1} \otimes   
\ket{g(2i+2)}_{2i+2} \otimes 
\cdots,
\end{align}
where 
any entry  on  the complement of  $\{2i-1, 2i, 2i+1\}$ in $\Z$
is  unchanged.
For any other $g(n)$  the  corresponding 
vector $\ket{g(n)_{n\in\Z}}$ is always deleted:
\begin{align}
\label{eq:chanicast-act-zero}
\chanictwoisanast
\ket{g(n)_{n\in\Z}}  
=0.
\end{align}
The above  relations  
\eqref{eq:chanic-act-zero}
\eqref{eq:chanicast-act-zero}
have established    that   all 
local fermion charges of the Nicolai model  
 delete   any  classical vector 
$\ket{g(n)_{n\in\Z}}$ if  
  $g(n)$ is a ground-state configuration. 
Thus  we have shown the if part of the statement.

We will   show the only if part of the statement.
Suppose  that 
 a classical supersymmetric state $\psi_{g(n)}$ is given, where
 $g(n)$ is its associated configuration over $\Z$.
By the assumption   both  
 $\Qnic \ket{g(n)_{n\in\Z}}=0$ 
 and  $\Qnicast \ket{g(n)_{n\in\Z}}=0$ hold.
(Note that the existence of these  supercharge   operators
 on the GNS Hilbert space  for any  supersymmetric state  
 is guaranteed \cite{AHP16}.)
From 
\eqref{eq:chanic-act-nonzero}
\eqref{eq:chanic-act-zero}
the  identity   $\Qnic \ket{g(n)_{n\in\Z}}=0$
 implies that $\chanictwoisan \ket{g(n)_{n\in\Z}}=0$ 
 for all    $i \in \Z$,
since there is no
cancellation   among  the actions of the local fermion charges 
 upon  $\ket{g(n)_{n\in\Z}}$. 
Similarly, from 
\eqref{eq:chanicast-act-nonzero}
\eqref{eq:chanicast-act-zero} the identity 
 $\Qnicast \ket{g(n)_{n\in\Z}}=0$ implies that 
 $\chanictwoisanast \ket{g(n)_{n\in\Z}}=0$   for all  $i \in \Z$.
These facts  imply  that $g(n)$ should  be  a ground-state configuration
 in which   no forbidden triplet is  included.

By Proposition  \ref{prop:SUSYground}, the invariance  
  under the time evolution is obvious.
\end{proof}

We now  provide  another remarkable characterization 
 of  the classical supersymmetric ground states of the Nicolai model.
 For this purpose   we  recall the  formula 
  $\Hnic=H_{\rm{classical}}+H_{\rm{hop}}$  \eqref{eq:Hnic-decom}, where 
$\Hnic$ is the total Hamiltonian  given explicitly in 
\eqref{eq:Hnicgutai}, $H_{\rm{classical}}$ is  the classical term given 
in \eqref{eq:Hnic-cla}, and  $H_{\rm{hop}}$ is the hopping term given in  
\eqref{eq:Hnic-hop}. We consider  a new  classical spin lattice model 
 determined by  $H_{\rm{classical}}$ (which is imbedded in the fermion lattice system).
\begin{theorem}
\label{thm:skelton}
The set of  all classical supersymmetric 
 ground  states for the Nicolai model over  $\Z$ is identical to the set of 
all ground states for the  classical spin model over $\Z$ corresponding 
to the  classical part   of the Nicolai model{\rm{:}}\\ 
$H_{\rm{classical}}
=\sum_{i\in \Z}n_{2i}-n_{2i-1}n_{2i}
-n_{2i}n_{2i+1} +n_{2i-1}n_{2i+1}$.
\end{theorem}

\begin{proof}
Take any three-site  subset  $\{2k-1, 2k, 2k+1\}$ 
centered  at an even site $2k$ ($k\in\Z$).
 There are    eight ($=2^3$) classical   configurations  on   
$\{2k-1, 2k, 2k+1\}$.  The  local interaction   within  
 $\{2k-1, 2k, 2k+1\}$ 
is $m_{2k}:=n_{2k}-n_{2k-1}n_{2k}-n_{2k}n_{2k+1}+n_{2k-1}n_{2k+1}$.
 The operator $m_{2k}$  takes eigenvalue $+1$
  upon  the  two forbidden triplets  
  $``0, 1, 0"$  and  $``1, 0, 1"$ on $\{2k-1, 2k, 2k+1\}$, while it 
 takes $0$  on  the other six  classical  configurations on 
$\{2k-1, 2k, 2k+1\}$.
As  $H_{\rm{classical}}$ 
is the summation  of these positive operators  $m_{2k}$, it is positive.
  $H_{\rm{classical}}$  takes 
   $0$ on any  classical configuration that does 
 not include the   forbidden triplets  
  $``0, 1, 0"$   and  $``1, 0, 1"$ anywhere over $\Z$, 
 while  it takes  a strictly positive value on any other classical configuration.    
Thus  $H_{\rm{classical}}$ 
takes its minimum value $0$ only on the ground-state configurations  
defined  in  Definition \ref{defn:ground-config}.
Therefore by Theorem \ref{thm:Allclassic} we obtain  the equivalence  as  stated. 
\end{proof}

\begin{remark} 
\label{rem:KMN}
 The  highly degenerated 
classical ground states shown  in this section
 can be understood as symmetry breakdown 
of  hidden local fermion symmetries  given   in 
$\S$\ref{sec:INFINITE-CONSTANT}. 
 For the detail we refer  to \cite{KMN}.
\end{remark}

\section{Ergodicity breaking}
\label{sec:ERGODIC}
From the existence of local constants of motion  shown in 
 $\S$\ref{sec:INFINITE-CONSTANT} 
  we immediately see that the Nicolai model breaks ergodicity.
In this section we shall show  several non-ergodic properties  of the Nicolai model. 

In $\S$\ref{subsec:ERGODIC}   we  consider   the 
  ergodicity  due to  Mazur \cite{MAZUR} 
 which  is  given  in terms of  averaged temporal autocorrelation functions of 
 invariant states.  We  prove that the Nicolai model breaks ergodicity in this sense.
In $\S$\ref{subsec:ROECK}  we  see   that  
  delocalization of the Hamiltonian dynamics of the Nicolai model is suppressed 
 although it  has a non-trivial disorder-free
translation-invariant  interaction.
In $\S$\ref{subsec:INTEGRABILITY} we investigate  quantum integrability 
 that  the Nicolai model possesses
  based on the proposal by Caux-Mossel \cite{CAUX}.

\subsection{Ergodicity breaking in the sense of Mazur}
\label{subsec:ERGODIC}
We   recall the definition of ergodicity 
 due to Mazur \cite{MAZUR}  for  a general $\cstar$-dynamical system 
  as  stated in  \cite{SIR-VER73}.
Consider  a one-parameter group  of automorphisms $\alpha_t$ ($t\in\R$) 
 on a $\cstar$-algebra $\Al$. 
Assume that $\alpha_t$ ($t\in\R$) be  strongly continuous: 
\begin{equation}
\label{eq:strongalt}
\lim_{t \to 0}\Vert\alt(A)-A\Vert\ \to 0
\ \ {\mbox{for every}}\  A \in \Al.
\end{equation} 
Suppose that  a state $\ome$  on $\Al$ is $\alt$-invariant,
\begin{equation}
\label{eq:omealtinv}
\ome\left(\alpha_t(A)\right)=\ome (A) 
\ \ {\text{for all}}\ A\in \Al \ \ {\text{and}}\ t\in\R. 
\end{equation}
 The triplet $(\Al, \alpha_t, \ome)$ 
is   called a quantum dynamical system ($\cstar$-dynamical system).

By  $\bigl(\Hilome,\; \piome,\; \Omeome  \bigr)$
 we denote  the  GNS representation  associated to the state  $\ome$ of $\Al$.
 Precisely,   
$\piome$ is  a  homomorphism from $\Al$
 into $\Bl(\Hilome)$ (the set of all bounded linear 
 operators on the Hilbert space $\Hilome$), 
 and  $\Omeome\in\Hilome$ is a cyclic vector such that 
$\ome(A)=\left(\Omeome, \piome(A)\Omeome\right)$  for all $A\in \Al$.

 By  the  continuity  of  $\alpha_t$ with respect to $t\in\R$, 
   there exists  a strongly continuous  unitary   group 
 $\{\Utome;\; t\in\R \}$ that implements $\alpha_t$ ($t\in\R$)
on  the GNS Hilbert space $\Hilome$ as:  
\begin{align}
\label{eq:weakdynamics}
\Utome \left( \piome (A) \right)\Utome^{-1}&=\piome\left(\alpha_t(A)\right)
\ \ {\text{for all}}\ A\in \Al \ \ {\text{and}}\ t\in\R. 
\end{align}
By the Stone-von Neumann theorem \cite{RSI}, 
 there exists  a self-adjoint operator $\Home$ on $\Hilome$ such  that 
\begin{equation}
\label{eq:Utome}
\Utome=e^{it \Home} \ \ {\text{for all}} \ t\in \R,
\end{equation}
 and 
\begin{equation}
\label{eq:Hann}
\Home\Omeome=0.
\end{equation}
Let $F_{\ome}$ denote the orthogonal projection 
 on the $\Utome$-invariant vectors in $\Hilome$, i.e 
 the projection in $\Hilome$ with the range 
\begin{equation}
\label{eq:Frange}
\bigl\{ \psi\in \Hilome \; |\; \Utome\psi=\psi   \ \ {\text{for all}} \ t\in \R\bigr\}.
\end{equation}

With the above notations in hand, we shall introduce the notion of ergodicity.
 The following inequality holds 
for any $A\in \Al$ as shown in   \cite{SIR-VER73}.
\begin{align}
\label{eq:generalinequ}
\lim_{T\to\infty}\frac{1}{T} \int_{0}^{T} 
 \ome \left(A^{\ast}\alpha_t(A) \right)dt
 \ge \ome(A^{\ast})\ome(A).
\end{align}
 The operator $A$ is called ergodic
 if this becomes an  equality.   
Otherwise, $A$ is called a non-ergodic operator.
If every operator  of $\Al$ is   ergodic, then 
the quantum dynamical system   $(\Al, \alpha_t, \ome)$  
is called  ergodic. Otherwise
 $(\Al, \alpha_t, \ome)$ is called  non-ergodic.

Our precise statement  of  non-ergodicity of the Nicolai model is as follows.
\begin{theorem}
\label{thm:ERGO-break-NIC}
Let $\alchanict$ ($t\in \R$) denote the 
 time evolution 
 of  the Nicolai  model given  in $\S$\ref{subsec:NIC}.
 For any KMS state  with respect to $\alchanict$ ($t\in \R$)
at any positive temperature $\beta\in\R$, 
 the ergodicity in the sense of Mazur is broken.
 For any classical supersymmetric ground state  
 given in Theorem \ref{thm:Allclassic} the ergodicity is also broken.
\end{theorem}

\begin{proof}
In Theorem 2  of  \cite{SIR-VER73}   the 
 criterion of ergodicity is established as follows:\ 
{\it{The quantum dynamical system  $(\Al, \alpha_t, \ome)$
 is ergodic if and only if  
$F_{\ome}$ is a one-dimensional projection.}}

First we consider  KMS states.
 For the precise definition  of KMS states, 
 we refer to  \cite{BRI,BRII}.
Let $\ome$ denote a KMS with respect to 
 $\alchanict$ ($t\in \R$) (whose existence has been known).  
Take  any  $\alchanict$-invariant element  $B\in\CA$.
Assume that  it  is not a scalar. 
Actually there are many non-scalar elements in  $\CA$ 
 by   Definition \ref{defn:CA}.
 As any  KMS state is known to be a faithful state, 
 $\ome(B^\ast B)$ is strictly positive, and 
hence $\piome (B)\Omeome  \ne 0$.
 We   normalize $B$ so that  $\Vert \piome (B)\Omeome\Vert=1$.
 Let us denote this new   normalized vector $\piome (B)\Omeome$ by  $\Omeome^{B}$. 
 It is known that the GNS vector $\Omeome$ of any KMS state  
is a  separating vector. Thus  
 $\Omeome^{B}$ and $\Omeome$ are different rays that give rise to different  states, 
 namely  $\Omeome^{B}\ne \Omeome$ up to $U(1)$-phase. 
As $\alchanict(B) =B$ and $\Utome^{-1} \Omeome=\Omeome$ 
for all $t\in \R$, we see  
\begin{align*}
\Utome\Omeome^{B} =\Utome\piome(B)\Omeome =
\Utome\piome(B)\Utome^{-1} \Omeome\\
=\piome(\alchanict(B)) \Omeome
=\piome(B) \Omeome=\Omeome^{B}.
 \end{align*}
This tells that  $\Omeome^{B}$ is in   the range of $F_{\ome}$.
 Hence the range of  $F_{\ome}$ has more than one-dimension.
 Thus  $(\Al, \alchanict, \ome)$ is   non-ergodic.

Second we consider the case of classical supersymmetric ground  states.
Let $\ome$ denote any  such  state given in  
Theorem \ref{thm:Allclassic}. 
By Theorem \ref{thm:Allclassic} 
 there are many other ground states  which are  identical to $\ome$ 
 except on some finite region. 
Namely there are infinitely degenerated  ground states in the same Hilbert space 
 $\Hilome$.
Therefore the range of $F_{\ome}$ has more than one-dimension (in fact infinite dimension).
 We conclude that  $(\Al, \alchanict, \ome)$ is   non-ergodic.   
\end{proof}

\begin{remark} 
\label{rem:various}
There are some  different  definitions of ergodicity  
   in addition to 
 the definition  \cite{SIR-VER73} which  we have chosen here. 
See e.g. \cite{BRII}  \cite{THIR4} for   mathematical  formalism based on  operator algebras
 and  \cite{POLK}  for more  physics oriented treatment. 
If a strong chaotic  property of dynamics  known as the  asymptotic abelian condition \cite{DKR66} 
is satisfied,  then a straightforward  quantum  generalization of  classical ergodic  theory  
 is possible  as noted in \cite{THIR4}. 
However, the  asymptotic abelian condition  remains  an unjustified hypothesis  \cite{NAR15}.
 In fact, it is violated for  the Nicolai model.
\end{remark}

\subsection{Failure  of  delocalization}
\label{subsec:ROECK}
In  \cite{Deloc} 
a  general scenario of delocalization for 
  disorder-free translation-invariant quantum Hamiltonians is proposed.
 We will apply   this  scenario  
 to  the Nicolai model for  some natural but restricted  case.
 To  this end   we  recall that  the Nicolai model has a  decomposition 
  $\Hnic=H_{\rm{classical}}+H_{\rm{hop}}$  \eqref{eq:Hnic-decom}, where 
$H_{\rm{classical}}$ is  the classical term given in \eqref{eq:Hnic-cla}, and  $H_{\rm{hop}}$ is the hopping term given in  \eqref{eq:Hnic-hop}. 

\begin{proposition}
\label{prop:NORESONANT}
Take  $H_{\rm{classical}}$ as  our  initial classical Hamiltonian.
Consider its perturbation by the  quantum interaction $H_{\rm{hop}}$.
 Then all   ground states of $H_{\rm{classical}}$ (that exist 
 infinitely many)  are invariant under any order  of the  perturbation 
 by  $\lambda H_{\rm{hop}}$ ($\lambda\in\R$).
\end{proposition}

 \begin{proof}
By   Proposition \ref{prop:corresonFock},  
 Theorem \ref{thm:Allclassic} 
 and Theorem \ref{thm:skelton}, 
 any ground state of $H_{\rm{classical}}$ 
is represented  by a  vector 
$\ket{g(n)_{n\in\Z}}$ with some $g(n)\in \Upsilon$ defined in 
   Definition  \ref{defn:ground-config}.
 As it is a ground state for both
 $\Hnic$ and  $H_{\rm{classical}}$, the following identities hold: 
\begin{align}
\label{eq:actzero}
\Hnic \ket{g(n)_{n\in\Z}}=0=H_{\rm{classical}} \ket{g(n)_{n\in\Z}}.
\end{align}
From \eqref{eq:actzero} and $H_{\rm{hop}}=\Hnic-H_{\rm{classical}}$
 we have 
\begin{align}
\label{eq:classiact}
H_{\rm{hop}} \ket{g(n)_{n\in\Z}}=0.
\end{align}
This implies that  for any $k\in \NN$ and any $\lambda\in\R$
\begin{align}
\label{eq:classiact}
{(\lambda H_{\rm{hop}})}^k  \ket{g(n)_{n\in\Z}}=0.
\end{align}
So we obtain non existence of resonance as  
\begin{align}
\label{eq:noresonant}
\bra{\psi} {(\lambda H_{\rm{hop}})}^k  \ket{g(n)_{n\in\Z}}=0, 
\end{align}
where  $\psi$ is any state. 
\end{proof}

\begin{remark} 
\label{rem:notharm}
 The high  degeneracy of ground states 
  is   not  harmful for delocalization;  
 this   would  even  make   resonance happen  easier. 
The model  given in \cite{Deloc} is a generic 
interacting boson lattice model, whereas our model is a fermion lattice model. 
 As    the (spinless) fermion lattice model has much fewer 
 degrees of freedom at each site (only  up and down)
 than   boson  models,   more  resonant spots will happen. 
\end{remark}

\begin{remark} 
\label{rem:scenarioII}
Our statement  does not invalidate   the {\it{generic}} scenario  of delocalization 
 considered  in  \cite{Deloc}. 
The  formula  \eqref{eq:noresonant}  tells  no-resonant
  for the  particular  perturbation  upon  only classical ground  states. 
Resonant  may happen if we take  other 
  quantum hopping perturbations. 
\end{remark}

\subsection{On Quantum integrability}
\label{subsec:INTEGRABILITY}
We  shortly discuss   quantum integrability  for  the Nicolai model
 based  on  the  definition by  Caux-Mossel \cite{CAUX}.
Their  definition of {\it{quantum integrability}} 
  consists of four requirements which are  referred to as  
Requirements 1 to 4. Below we will check  them for the Nicolai model.

The set of nilpotent equations \eqref{eq:Qseq-nil} 
and the anti-commutation relations    
\eqref{eq:anticommute} given in  Proposition 
\ref{prop:ALGQseq} 
 will correspond  to the first half of Requirement 1.
 (The original definition  is designed  for  bosonic (usual) symmetries. 
Here  we  replace   the commutator 
 by the anti-commutator as   we deal with   fermionic  symmetries.) 
 By   \eqref{eq:easyQfconst}  in 
Theorem \ref{thm:NICconstant},   any of 
  $\{\Qseq(f)\in \core|\; f\in \hat{\Xi} \}$ commutes 
   with the total Hamiltonian $\Hnic$.  
This corresponds to  the second half of Requirement 1.
 So we have verified  Requirement 1 for the  set of local constants of motion 
 generated by the   local fermion operators  
$\{\Qseq(f)\in \core|\; f\in \hat{\Xi} \}$.

From Theorem \ref{thm:NICconstant} and 
 Remark \ref{rem:crude} one sees that 
  the number of  the set of local  constants of motion
 $\{\Qseq(f)\in \core|\; f\in \hat{\Xi} \}$
 increases   exponentially   with respect to the volume of  subsystems.
However, counting  independent operators  needs some care.
 In fact all the operators in 
  $\{\Qseq(f)\in \core|\; f\in \hat{\Xi} \}$ are not algebraically independent.
(By using  the CAR relations,  one can verify  that 
 the operators in  $\bigcup \{\Qseq(f) |\; f\in \hat{\Xi}_{0,\; 1}\bigcup \hat{\Xi}_{0,\; 2}  \}$
 are algebraically independent. However, the operators in  
$\bigcup \{\Qseq(f) |\; 
f\in \hat{\Xi}_{0,\; 1}\bigcup  \hat{\Xi}_{0,\; 2} \bigcup  \hat{\Xi}_{0,\; 3}   \}$  
are not algebraically independent.)
In any case,  we can see that  Requirement 2 is satisfied by  the similar  reason 
for  the free theories  as described in Sec.5 of \cite{CAUX}.

As  the cardinality of   $\{\Qseq(f)\in \core|\; f\in \hat{\Xi} \}$ 
 is unbounded,  Requirement 3 is satisfied.   

Requirement 4 is rather involved, 
so  we refer the readers 
 to the original paper \cite{CAUX}.
  We will only indicate  essential  points.
Any operator $\Qseq(f)\in \core$ for $f\in \hat{\Xi}$  
is a monomial 
 of finite fermion creation and  annihilation operators.
 Hence   it has the  constant character 
 of the preferred basis (the Fock-state  basis as in Definition \ref{defn:CLASSIC}).  
Hence   Requirement 4 is satisfied.  Our conclusion is now stated as follows:
\begin{proposition}
\label{prop:Caux}
The Nicolai model belongs to the 
 constant class of quantum  integrability 
 in the sense of   Caux-Mossel.
\end{proposition}

\begin{remark} 
\label{rem:Caux}
 Proposition \ref{prop:Caux} is valid for 
  the    Nicolai model on any dimensional integer lattice,   
  see $\S$\ref{subsec:Multi}. 
\end{remark}

\begin{remark} 
\label{rem:CauxII}
 It is arguable that  
  the non-interacting  models  
 and   the Nicolai model    belong  to the same  class  
(the  constant class) of quantum integrability.
 Requirement 3   merely  requires    infinite  number of  independent local constants of motion.
 However,  the completeness  of  such   local constants of motion 
 is to  be taken into account for  more  precise  characterization, see
  \cite{RADE}  \cite{WEIGERT}
 \end{remark}

\section{Generalization}
\label{sec:gen}
We have studied the Nicolai  model 
  on one-dimensional lattice $\Z$. 
In this section  we shall  discuss  generalizations
 of  the  results   given  so far.
 We  discuss  generalization to 
  multi-dimensional lattice in $\S$\ref{subsec:Multi}. 
We consider the Nicolai model for finite systems 
in $\S$\ref{subsec:finite}.

\subsection{The Nicolai model on  multi-dimensional lattice}
\label{subsec:Multi}
We  will  indicate that 
  our results   given  so far  can be easily extended to
 the Nicolai model on $\Znu$ of  arbitrary $\nu\in\NN$.
 In the following we  discuss   $\Z^{2}$ since this  essentially 
 represents all the cases of  $\Znu$.

First we set up a   model for   $\Z^{2}$.
Let us define 
\begin{align}
\label{eq:Qnictwo}
\Qnic&:=
\sum_{(i,j) \in \Z^2}  
\chanictwoij,\\\nonumber  
\chanic(\twoij)&:= 
\atwoimtwoj \atwoitwojm \atwoitwojcr  
\atwoiptwoj \atwoitwojp \in \Altwoijo,
\end{align}
where 
\begin{align}
\label{eq:twoij}
\twoij:=\twoitwoj  \quad {\text{for}}\ i,j\in\Z. 
 \end{align}
 Namely the finite subset $\twoij\Subset\Z^{2}$ consists of five sites,  its  center  
$(2i,2j)\in (2\Z)^{2}$ and four sites next to  the center.  
It is easy to see that $\Qnic$ is  nilpotent, i.e. $\Qnic^2=0$.
From  the above $\Qnic$ we can construct a supersymmetric $\cstar$-dynamics 
 on  the  fermion lattice system  $\Al$ over $\Z^2$
 in much the same way as  given in $\S$\ref{subsec:NIC}.
 Proposition \ref{prop:AHPmain} holds  for the  two-dimensional Nicolai model.

As  in  $\S$\ref{sec:INFINITE-CONSTANT} 
 we can  construct  infinitely many local fermionic  constants  of  
 this new time evolution $\alchanict$ ($t\in \R$).
To this end  we  replace 
 Definition \ref{defn:forbid-permit} for $\Z$ by the  following one.
\begin{definition} 
\label{defn:newforbid-permit}
Let  $\I$ be  any   rectangle   of $\Z^{2}$.
Let $f$ be a $\{-1, +1\}$-valued function  on  $\I$.
If  on some $\twoij\Subset \I$,  either     
\begin{align}
\label{eq:forbidII}
&f((2i,2j))=+1,\nonumber \\
&f((2i-1,2j))=-1,\; f((2i,2j-1))=-1,\;    
f((2i+1,2j))=-1, \; f((2i,2j+1))=-1,  
\end{align}
or 
\begin{align}
\label{eq:forbidII-ura}
&f((2i,2j))=-1,\nonumber \\
&f((2i-1,2j))=+1,\; f((2i,2j-1))=+1,\;  
f((2i+1,2))=+1,\; f((2i,2j+1))=+1
\end{align}
is satisfied, then  $f$ is called   forbidden.
Otherwise, $f$  is called  permitted.
\end{definition}

With this new definition, we can immediately generalize 
 Definition \ref{defn:seq-cons} to the case of $\Z^{2}$ and  obtain 
local configurations  of conservation  for the Nicolai model on $\Z^{2}$.
Namely we can set up an  analogous rule  
given in $\S$\ref{subsec:LCS} for  $\Z^{2}$ and   
 provide  local fermionic operators  in the same way.
 We shall give  an example.
Let 
\begin{align} 
\label{eq:rectlm}
\Izerolm \equiv\{(x,y)\in\Z^{2};\; 
 0\le x  \le 2l,\  
0\le y  \le 2m\}\quad (l,m \in \NN).
\end{align}
As in Definition   \ref{defn:ID-conserved} take the 
 simplest local {{configurations}}:
\begin{equation}
\label{eq:constkl}
r^{+}_{[0,2l]\times [0,2m]}
(i)=+1 \  \forall i\in \Izerolm,\quad    
r^{-}_{[0,2l]\times [0,2m]}
(i)=-1 \  \forall i\in \Izerolm.
\end{equation}
The  assignment
 of local fermion operators 
 from local {{configurations}}   of conservation 
 will be denoted by the same symbol  $\Qseq$ as in 
  Definition \ref{defn:local-assignment}.
Then we have 
\begin{align}
\label{eq:constmultibox}
\Qseq(r^{+}_{[0,2l]\times [0,2m]})
=\prod_{i\in \Izerolm}\!\!\!\!\!\! \aicr  \in {\Al}(\Izerolm)_{-},\nonumber\\
\Qseq(r^{-}_{[0,2l]\times [0,2m]})
=\prod_{i\in \Izerolm}\!\!\!\!\!\! \ai \in {\Al}(\Izerolm)_{-},
\end{align}
where we specify certain   order of products.
 Repeating almost the same argument  as 
 in  $\S$\ref{sec:INFINITE-CONSTANT} we can show that
  both of them  are   
 invariant under the time evolution $\alchanict$ ($t\in \R$).

From the above generalization 
we can easily  derive  similar  results 
   shown in   $\S$\ref{sec:C-G-States}  $\S$\ref{sec:ERGODIC} 
 for the Nicolai model  on $\Z^{2}$.

 \subsection{Finite-volume  models}
\label{subsec:finite}
Let $m\in 2\NN$. 
We shall introduce the Nicolai model 
 on $\Al([-m-1,m])$ under the   periodic boundary condition specified below.
(Of course one may choose other boundary conditions.)
Define 
\begin{align}
\label{eq:Qfinite}
\Qnic[-m-1,m]:=
\sum_{i=-m/2}^{m/2}  
\chanic(\{2i-1, 2i, 2i+1\}),\quad 
\chanic(\{2i-1, 2i, 2i+1\})\equiv  a_{2i+1} a^{\ast}_{2i} a_{2i-1},
\end{align}
where the site $m+1$ is  identified with the site  $-m-1$.
We see that 
\begin{align}
\label{eq:Qgper-nil}
\Qnic[-m-1,m]^2=0=\Qnic[-m-1,m]^{\ast\;2}.
\end{align}
Thus the nilpotent condition is satisfied. 
The finite-volume  supersymmetric Hamiltonian 
is naturally generated by  the above  finite supercharges  as 
\begin{equation}
\label{eq:Hgper}
\Hnic[-m-1,m]:=\Bigl\{\Qnic[-m-1,m],\; \Qnic[-m-1,m]^{\ast} \Bigr\}.
\end{equation}

It is convenient to introduce the following 
intervals whose  edges are  odd  (c.f. \eqref{eq:Ikl}) 
\begin{align} 
\label{eq:Imtil}
\Imtil\equiv[-m-1, m+1]
 \quad \text{with identification}\    m+1\equiv -m-1.
\end{align}
 We take  ${\Xi}_{\Imtil}$  as  in Definition \ref{defn:forbid-permit}, i.e.  
the set of all $\{-1, +1\}$-valued permitted sequences on  $\Imtil$. 
Then as in Theorem \ref{thm:NICconstant} we obtain 
\begin{equation}
\bigl\{\Qnic[-m-1,m],\; \Qseq(f) \bigr\}=0=\left\{\Qnic[-m-1,m]^{\ast} \;  \Qseq(f)\right\} 
\ \ \text{for every}\ f\in {\Xi}_{\Imtil}.
\end{equation}
This implies 
\begin{equation}
\Bigl[\Hnic[-m-1,m], \;  \Qseq(f)\Bigr]=0 \ \ \text{for every}\ f\in {\Xi}_{\Imtil}.
\end{equation}
Namely every  $\Qseq(f)\in \Al([-m-1,m])$  with $f\in {\Xi}_{\Imtil}$ gives  a constant of motion
 for the Heisenberg time evolution generated by the finite-volume   Hamiltonian $\Hnic[-m-1,m]$.

We can provide all classical ground states for  the  Hamiltonian $\Hnic[-m-1,m]$
 as in  Theorem \ref{thm:Allclassic}
 in terms of   the  set of all ground-state  configurations
in $\Imtil$ given  in  Definition \ref{defn:ground-config}.

The finite-volume Nicolai model on $\Z^2$ can be given similarly.

\section{Summary and Discussion}
\label{sec:DIS}
 We have studied  dynamics of the Nicolai supersymmetric  fermion lattice model.
We have given explicitly  its  infinitely many 
  local fermionic constants of  motion.
The number of these local constants is extensive.
 As a consequence of them,  ergodicity breaking and certain  many-body localization 
 manifest.

All 
 classical supersymmetric ground states 
for the Nicolai model are determined in terms of  classical configurations (binary codes).
 We may interpret these   infinitely many   classical ground states on the Fock space 
as  a  many-body localization phenomena.  
Recently, 
 ergodicity  and its breakdown  have been discussed 
 in  the   subject  ``thermalization and  many-body  localization (MBL)''
 \cite{HUSErev} \cite{POLK}.
It is  believed  that many-body  localization is  essentially  caused by  strong 
  disorder,  and this belief has been verified 
 for some  models \cite{BASKO}  \cite{IMB}. 

We shall discuss the Nicolai model which  has a clean Hamiltonian with no disorder
 being  inspired by  the following   question posed  in 
  \cite{SCHIULAZ15} \cite{SMITHA-dfree}:
 Does MBL always necessitate  disorder? 

Usually  many-body localization (MBL)   requires 
   localization for  {\it{all}} eigenstates
 of interacting quantum models  \cite{OGANE-HUSE07}.
(See \cite{EMERGENT} where a weaker notion of  MBL 
  is proposed.) 
 On the other hand, we have shown  merely localization at  zero temperature for the Nicolai model.
Furthermore,  the Nicolai model  lacks a {{complete}}  set of ``l-bits'' (local integrals of motion)
 which seems to be  considered   essential for MBL  
  \cite{CHANDRAN}  \cite{PHENO} \cite{IMB-ROS-SCARD} 
   \cite{NANDURI-KIM} \cite{ROS15}  \cite{SERBYN13LCL}. 

 We now   provide    a   heuristic derivation of the lack of 
  a {{complete}}  set of l-bits for the Nicolai model.
 By Theorem \ref{thm:Allclassic}
 the Nicolai model has  infinitely many  classical ground states 
 which  are all  product states  with respect to the set $\{n_i;\; i\in\Z\}$.
As  a complete set of classical l-bits 
 is to  be 
 uniquely determined by  Hamiltonian eigenstates  \cite{HUSErev}, 
   $\{n_i;\; i\in\Z\}$  should  be a desired  set of complete 
 l-bits for the Nicolai model.  (Those  are actually  the p-bits for  the fermion system.)  
 Obviously  $H_{\rm{classical}}$ is  diagonalized with respect to 
 $\{n_i;\; i\in\Z\}$,  whereas   the hopping term 
 $H_{\rm{hop}}$ can  not be.
Hence $\Hnic$ can  not be  diagonalized with respect to  $\{n_i;\; i\in\Z\}$. 
   It is now concluded that  $\{n_i;\; i\in\Z\}$ is not a   desired  set of complete  l-bits. 
In other words, the l-bits 
$\{n_i;\; i\in\Z\}$ determined by classical  ground states 
do not coincide with    l-bits (constants of motion)
 determined by  the time evolution.

Finally we shall propose   some future  problems on the dynamics of the Nicolai model. 
\begin{enumerate}

\item Determine all  constants of motion for the time evolution of 
 the Nicolai model.

\item Determine concrete (quasi-)local  operators which  
 are ergodic  for the time evolution of the Nicolai model. 
 Those are state-dependent  \cite{CAUX} \cite{ZOTOS}.

\item Determine whether  the time evolution of the Nicolai model 
 has chaotic properties or it is completely frozen. 
 (What is the value of dynamical entropy for 
 the time evolution with respect to invariant states \cite{MOHYA}?)

\item  Consider   quench  dynamics  of   the Nicolai model.
 We may refer to  \cite{CUBERO-MUSS}.

\item  How does  a generalized Gibbs ensemble 
  \cite{RIGOL-GGE}  look like?
A notable point with  the Nicolai model 
 is that it has  many fermionic  constants  of motion.

\item  We  may speculate   a randomized    Nicolai model
 by changing  its   coefficients  by  random variables 
 as the  supersymmetric  Sachdev-Ye-Kitaev model \cite{FU}.

\end{enumerate}

\begin{acknowledgements}
I  thank 
 Prof. Deguchi, Prof. De Roeck, Prof. Huveneers,
Prof.  Katsura,  Dr. Padmanabhan, Prof. Nakayama and  Prof. Narnhofer 
 for  discussion. 
It is with great honor that I mention my memory  of  Prof. Masanori Ohya (1947-2016)
who developed Information Dynamics and conducted interdisciplinary studies.
Prof. Ohya  frequently organized  conferences  in Japan   where I  
encountered  various topics  and many  researchers. The  
scientific  interactions that I had with Prof. Ohya 
 helped motivate my investigation into quantum dynamical systems.
\end{acknowledgements}



\end{document}